\documentclass[11pt]{article}
\usepackage{amssymb}
\usepackage[all]{xy}
\usepackage{graphicx}

\title{The Logos Categorical Approach to Quantum Mechanics: I.\\Kochen-Specker Contextuality and Global Intensive Valuations.}

\author{{\sc C. de Ronde}$^{1,2}$ and {\sc C. Massri}$^{3,4}$}
\date{}

\usepackage[margin=2.5cm]{geometry}

\begin{document}

\bibliographystyle{plain}
\maketitle

\begin{center}
\begin{small}
1. Philosophy Institute Dr. A. Korn (UBA-CONICET) - UNAJ\\
2. Center Leo Apostel for Interdisciplinary Studies\\ Foundations of the Exact Sciences (Vrije Universiteit Brussel)\\
3. Institute of Mathematical Investigations Luis A. Santal\'o (UBA-CONICET)\\
4. University CAECE
\end{small}
\end{center}

\begin{abstract}
\noindent In this paper we present a new categorical approach which attempts to provide an original understanding of QM. Our logos categorical approach attempts to consider the main features of the quantum formalism as the standpoint to develop a conceptual representation that explains what the theory is really talking about ---rather than as problems that need to be bypassed in order to allow a restoration of a classical ``common sense'' understanding of {\it what there is}. In particular, we discuss a solution to Kochen-Specker contextuality through the generalization of the meaning of {\it global valuation}. This idea has been already addressed by the so called topos approach to QM ---originally proposed by Isham, Butterfiled and D\"oring--- in terms of {\it sieve-valued valuations}. The logos approach to QM presents a different solution in terms of the notion of {\it intensive valuation}. This new solution stresses an ontological (rather than epistemic) reading of the quantum formalism and the need to restore an objective (rather than classical) conceptual representation and understanding of quantum physical reality.
\end{abstract}
\begin{small}

{\bf Keywords:} {\em Categories, Logoi, contextuality, Kochen-Specker theorem, intensive valuation.}
\end{small}

\newtheorem{theo}{Theorem}[section]
\newtheorem{definition}[theo]{Definition}
\newtheorem{lem}[theo]{Lemma}
\newtheorem{met}[theo]{Method}
\newtheorem{prop}[theo]{Proposition}
\newtheorem{coro}[theo]{Corollary}
\newtheorem{exam}[theo]{Example}
\newtheorem{rema}[theo]{Remark}{\hspace*{4mm}}
\newtheorem{example}[theo]{Example}
\newcommand{\proof}{\noindent {\em Proof:\/}{\hspace*{4mm}}}
\newcommand{\qed}{\hfill$\Box$}
\newcommand{\ninv}{\mathord{\sim}} 
\newtheorem{postulate}[theo]{Postulate}

\bigskip

\bigskip

\bigskip

\bigskip

\bigskip

\bigskip

\bigskip

\section{Quantum Contextuality and the Representation of Reality}

Quantum contextuality presents one of most difficult problems in order to consider, within the orthodox formalism of QM, an objective representation of physical reality. If the choice of a context is necessary to define what is considered to be real, relativism enters the scene in a manner absolutely foreign to classical physical description. Objectivity is lost when the choice of a context appears as a necessary element in the definition itself of what can be considered to be the definite valued properties of a physical system. Every context defines its own ``relative reality''. As a consequence, the realist presuppositions according to which reality is one, preexists to measurements, and must be described independently of human consciousness or particular choices of agents becomes highly problematic \cite{deRonde17b}. 

One of us has argued elsewhere \cite{deRonde16c} that there exist two different notions of quantum contextuality addressed ---and often confused--- within the foundational literature. This distinction of contextuality will be useful for the purposes of the present article. The first, due to Bohr, is an epistemic notion of contextuality which grounds itself in the classical representation of experimental arrangements and the so called `wave-particle duality'. The second notion of contextuality is related to the Kochen-Specker (KS) theorem and the impossibility to interpret, within the orthodox quantum formalism, projection operators as definite valued (preexistent) properties. This second understanding of contextuality relates to an ontic questioning about the formalism of the theory and its possible conceptual representation ---beyond its mere reference to measurement outcomes or mathematical structures.\footnote{For a detailed analysis and discussion of the relation between conceptual representation, mathematical structures and measurement outcomes see \cite{deRonde17a}.} Let us discuss these two different notions of contextuality in some detail.  

The epistemic approach ---considered in general terms--- implies a perspective within philosophy of physics grounded on the empirical observation of subjects, i.e., a focus on the measurement process and outcomes. This was the approach taken by Bohr in his reply to the EPR paper \cite{Bohr35, EPR}. His analysis introduced a new idea of contextuality in line with an epistemic viewpoint. According to this idea, instead of the instrument playing a purely passive role, as it would classically, it was argued that there is in QM a complex interplay between system and instrument which ``creates something new''. We call the Bohrian notion of contexuality ``epistemic'' for it is grounded on the explicit reference to measurement situations with classical apparatuses in which classical phenomena can be addressed by an experimenter. Bohr remarked in many occasions there is an ``indispensable use of classical concepts in the interpretation of all proper measurements, even though the classical theories do not suffice in accounting for the new types of regularities with which we are concerned in atomic physics.'' Furthermore, [{\it Op. cit.}, p. 7], ``it would be a misconception to believe that the difficulties of the atomic theory may be evaded by eventually replacing the concepts of classical physics by new conceptual forms.'' In this respect, we remark that Borhian contextuality does not make reference to the quantum formalism itself. Since Bohr attempted to understand QM as a rational generalization of classical physics \cite{BokulichBokulich}, the focus of his analysis was centered in the description of complementary classical experimental situations and measurement results (see also for discussion \cite{Bacc14, Howard94}). His main idea was that the effect of the measurement apparatus is to ``create a value'' (of an observable), which did not exist before the measurement interaction. As John Bell \cite[p. 35]{Bell87} would later explain, rephrasing Bohr:  ``The result of a `spin measurement', for example, depends in a very complicated way on the initial position {\it x} of the particle and on the strength and geometry of the magnetic field. Thus the result of the measurement does not actually tell us about some property previously possessed by the system, but about something which has come into being in the combination of system and apparatus.'' This is perfectly consistent with Bohr's remark that the most important lesson of QM was an epistemic one, namely, that we are not only spectators but also actors in the great drama of (quantum) existence. Bohr developed his {\it epistemic} notion of contextuality based on the incompatibility of (complementary) measurement situations and the subsequent representation of phenomena (in terms of either waves {\it or} particles), stressing that \cite{Bohr29}: ``We must, in general, be prepared to accept the fact that a complete elucidation of one and the same object may require diverse points of view which defy a unique description.'' 

We might comprise Bohr's understanding of contextuality in the following manner:\\

\noindent {\it {\bf Bohrian Contextuality:} Since contexts, understood as experimental situations described in terms of classical theories, are incompatible in QM, the same quantum object might require mutually complementary physical representations. The most paradigmatic example of this notion of contextuality is provided by the double-slit experiment which gives rise to the so called ``wave-particle duality''. According to Bohr, this experiment requires both a `wave representation' and a `particle representation'. One might then consider the quantum object as if it is either a `wave' or a `particle' depending on the measurement set up chosen by the experimenter in the lab.}\\

However, within the foundational literature there also exists a very different approach which has considered the question of contextuality taking a formal or ontic\footnote{We use the term `ontic' to make clear the distance with respect to the `epistemic' views. See for a detailed discussion \cite{deRonde16b, deRonde17a}. As it has been remarked by van Fraassen, while the question of the ontological (or metaphysical) interpretation of a theory is something of main importance for the realist, for the empiricist this question is one of relative concern. He says \cite[p. 242]{VF91} that an ontological interpretation responds to the questions of ``what would it be like for this theory to be true, and how could the world possibly be the way this theory says it is?''} ---rather than epistemic--- viewpoint with respect to the orthodox formalism of QM. This approach goes back to Erwin Schr\"odinger who explained most clearly, also in 1935, his worries about the limits of classical ontology when considering the quantum formalism:
\begin{quotation}
\noindent {\small ``[In QM,] if I wish to ascribe to the model at each moment a definite (merely not exactly known to me) state, or (which is the same) to {\small {\it all}} determining parts definite (merely not exactly known to me) numerical values, then there is no supposition
as to these numerical values {\small{\it to be imagined}} that would
not conflict with some portion of quantum theoretical assertions." \cite[p. 156]{Schr35}}
\end{quotation}

\noindent Simon Kochen and Ernst Specker continued this analysis of the quantum formalism in their famous article of 1967 \cite{KS}. Taking as a standpoint the orthodox formalism of QM they asked ---implicitly--- a realist (or ontological) question regarding the possible {\it representation} of the theory which has no epistemic reference whatsoever ---i.e., no explicit reference to classical experimental situations, no reference to subjects or agents nor any reference to measurement outcomes. {\it Would it be possible to consider projection operators as actual (definite valued) preexistent properties within the orthodox formalism of QM?} This question, which attempts to understand the mathematical formalism in conceptual terms (i.e., in terms of systems with preexistent properties), led them to a very interesting analysis which we now shortly recall. 

In QM the frames under which a vector is represented mathematically are considered in terms of orthonormal bases. We say that a set $\{\alpha_1,\ldots,\alpha_n\}\subseteq {\cal H}$ in an $n$-dimensional Hilbert space is an \emph{orthonormal basis} if $\langle \alpha_{i} | \alpha_{j} \rangle = 0$ for all $1 \leq i < j \leq n$ and $\langle \alpha_i|\alpha_i\rangle=1$ for all $i=1,\ldots,n$. A physical quantity is represented by a self-adjoint operator on the Hilbert space ${\cal H}$. We say that $\mathcal{C}$ is a \emph{context} if $\mathcal{C}$ is a commutative subalgebra generated by a set of self-adjoint bounded operators $\{P_1,\ldots,P_s\}$ on ${\cal H}$. Quantum contextuality, which was most explicitly recognized through the KS theorem \cite{KS}, asserts that a value ascribed to a physical quantity $P$ cannot be part of a global assignment of binary values but must, instead, depend on some specific context from which $P$ is to be considered. Let us discuss this in some detail. 

Physically, a global binary valuation allows us to define the preexistence of actual definite valued properties of a system; i.e., the reference to the existence of properties independently of particular measurement observations.\\  

\noindent {\it
{\bf Binary Valuation:} A binary valuation is an function to $\{0,1\}$.
}\\

\noindent Mathematically, a \emph{valuation} over an algebra $\mathcal{A}$ of self-adjoint operators on a Hilbert space, is a real function satisfying,

\begin{enumerate}
\item[1.] \emph{Value-Rule (VR)}: For any $P\in\mathcal{A}$, the value $v(P)$ belongs to the spectrum of $P$, $v(P)\in\sigma(P)$.
\item[2.] \emph{Functional Composition Principle (FUNC)}: For any $P\in\mathcal{A}$ and any real-valued function $f$, i.e. $v(f(P))=f(v(P))$.
\end{enumerate}

\noindent We say that the valuation is a \emph{Global Binary Valuation (GBV)} if $\mathcal{A}$ is the set of all bounded, self-adjoint operators. In case $\mathcal{A}$ is a context, we say that the valuation is a \emph{Local Binary Valuation (LBV)}. We call the mathematical property which allows us to paste consistently together multiple contexts of {\it LBVs} into a single {\it GBV}, {\it Binary Value Invariance (BVI)}. First assume that a {\it GBV} $v$ exists and consider a family of contexts $\{ \mathcal{C}_i \}_I$. Define the {\it LBV} 
$v_i:=v|_{\mathcal{C}_i}$ over each $\mathcal{C}_i$. 
Then it is easy to verify that the set 
$\{v_i\}_I$ satisfies the \emph{Compatibility Condition (CC)}, 
$$v_i|_{ \mathcal{C}_{i} \cap \mathcal{C}_j} =v_j|_{\mathcal{C}_i\cap \mathcal{C}_j}
,\quad \forall i,j\in I.$$

\noindent The {\it CC} is a necessary condition that must satisfy a family of {\it LBVs} in order to determine a {\it GBV}. We say that the algebra of self-adjoint operators is \emph{BVI} if for every family of contexts $\{ \mathcal{C}_i\}_I$ and {\it LBVs} 
$v_i: \mathcal{C}_i \rightarrow \mathbb{R}$ satisfying the \emph{CC}, there exists a {\it GBV} $v$ such that $v|_{\mathcal{C}_i}=v_i$.

If we have {\it BVI}, and hence, a {\it GBV} exists, this would allow us to give values to all magnitudes at the same time maintaining a {\it CC} in the sense that whenever two magnitudes share one or more projectors, the values assigned to those projectors are the same in every context. 

\begin{definition}
Let $\mathcal{H}$ be a Hilbert space and let $\mathcal{G}$ be 
the set of observables.
An \emph{Actual State of Affair} (ASA)
is a global binary valuation
$\Psi:\mathcal{G}\to\{0,1\}$ such that $\Psi(I)=1$ and
\[
\Psi(\sum_{i=1}^{\infty} P_i)=
\sum_{i=1}^\infty \Psi(P_i)
\]
for any piecewise orthogonal projections $\{P_i\}_{i=1}^{\infty}$.
Notice that if $\dim(\mathcal{H})=n<\infty$, then the previous
condition says that $\Psi(\alpha_1)+\ldots+\Psi(\alpha_n)=1$ for any orthonormal
basis $\{\alpha_1,\ldots,\alpha_n\}$ of $\mathcal{H}$, 
where we denote $\Psi(\alpha_i)$ to indicate $\Psi(|\alpha_i\rangle\langle \alpha_i|)$.
\end{definition}

The KS theorem, in algebraic terms, rules out the existence of an ASA when the dimension of the Hilbert space is greater than $2$ and thus, an interpretation of projection operators as {\it preexistent} properties becomes problematic. 
The following theorem is a topological adaptation of the KS theorem ---as stated in \cite[Theorem 3.2]{DF}--- to the case of contexts:\footnote{See also for a detailed analysis of contextuality in the topos approaach \cite{RFD14, Eva16}.}

\begin{theo} 
{\bf (Kochen-Specker)} If ${\cal H}$ is a Hilbert space of 
$\dim(\mathcal{H}) >2$, then an ASA is not possible.
\end{theo}

\begin{coro}[KS Contextuality] Given a vector in Hilbert, the multiple contexts ---considered in terms of bases or complete set of commuting observables--- define projection operators which cannot be interpreted as preexistent properties possessing definite 
and compatible binary values, 0 and 1.
\end{coro}

\noindent Let us remark that the definition of contexts within KS type contextuality is based ---unlike Bohr's contextuality, which defines contexts in terms of classical descriptions--- on the mathematical formalism of QM alone. The result is completely formal. In this respect, it is only at a later stage of analysis, when including the (ontological) interpretation of projection operators as properties that KS theorem can be read as an {\it ad absurdum} proof of the impossibility to interpret a vector in Hilbert space in terms of a preexistent {\it Actual State of Affairs} (ASA). Let us explain this in more detail. By an ASA we mean a closed system considered in terms of a set of actual (definite valued) properties which can be thought as a map from the set of properties to the $\{0,1\}$. Specifically,  an ASA is a function 
$\Psi: \mathcal{G}\rightarrow\{0,1\}$ from the set of properties to $\{0,1\}$ satisfying certain compatibility conditions (see above).  
We say that the property $P\in \mathcal{G}$ is \emph{true} if $\Psi(P)=1$ and 
$P\in \mathcal{G}$ is \emph{false} if $\Psi(P)=0$. The evolution of an ASA is formalized by the fact that the morphism $f$  satisfies $\Phi f=\Psi$. Diagrammatically, 
\[
\xymatrix{
\mathcal{G}_{t_1}\ar[dr]_{\Psi}\ar[rr]^f&&\mathcal{G}_{t_2}\ar[dl]^{\Phi}\\
&\{0,1\}
}
\]
Then, given that $\Phi(f(P))=\Psi(P)$, the truth of $P\in \mathcal{G}_{t_1}$ is equivalent to the truth of $f(P)\in \mathcal{G}_{t_2}$. This formalization comprises the idea that the properties of a system remain existent through the evolution of the system. The model allows then to claim that the truth or falsity of a property is independent of particular observations. Or in other words, binary-valuations are a formal way to capture the classical actualist (metaphysical) representation of physics according to which the properties of objects preexist to their measurement. From a realist perspective ---recalling an example given by Einstein to Pauli---, the moon has a position regardless of whether we choose to observe it or not. This is in fact the main presupposition of the realist stance, the idea that reality has an existence independent of particular observations. However, when restricted to classical physics the claim becomes even more specific. It relates to a particular representation of physical reality in terms of an ASA. Something which is true in the particular cases of classical physical formalisms, all of which possess a commutative mathematical structure. 

\begin{theo} 
{\bf (Binary Non-Contextuality)} Let $\Gamma$ be a classical phase space of any dimension. 
Then, there exists an Actual State of Affairs.
\end{theo}
\begin{proof}
Classical observables commute. Hence, an ASA is the same as a GBV, that is, 
a function to $\{0,1\}$.
\qed
\end{proof}\\

To end this section, let us remark some points which will be important to keep in mind for the analysis of the following sections. While Bohrian contextuality is strictly related to our ``classical image of the world'', KS contextuality is a purely formal statement regarding valuations about the orthodox formalism of QM and the limits of its possible ontological interpretation in terms of definite valued properties. There is no direct implication between these two very different notions of contextuality \cite{deRonde16a}. While the first implies a reductionistic understanding of QM with respect to the classical representation of physics, the latter presents a formal constraint for the projection operators of a quantum state with respect to global binary valuations. Unlike Bohr's contextuality, KS theorem remains silent about the possible development of non-classical conceptual representations that would allow us to understand QM in a radically non-classical manner. As we shall discuss in this paper, while the topos approach seems to be part of the orthodox line of research which attempts to ``bridge the gap'' between QM and our classical worldview, the logos approach takes as a standpoint the formalism of QM and stresses the need to develop an objective (non-reductionistic) conceptual representation of physical reality ---one which, in principle, might even be non-classical.

\section{The Topos Approach to QM}

The topos approach was originally proposed by Chirs Isham \cite{Isham97}, Jeremy Butterfield and Andreas D\"oring \cite{DoringIsham08, DoringIsham11, DoringIsham11b}. At its origin, in a series of four papers \cite{IshamButter1, IshamButter2, IshamButter3, IshamButter4}, Isham and Butterfield investigated in depth the KS theorem attempting to provide a ``neo-realist'' solution that would restore a classical understanding about what QM is really talking about. The main ideas put forward by this neo-realist topos approach were summarized  by D\"oring and Isham in the following passage:
\begin{quotation}
\noindent {\small ``When dealing with a closed system, what is needed is a realist interpretation
of the theory, not one that is instrumentalist. The exact meaning of ``realist'' is
infinitely debatable, but, when used by physicists, it typically means the following:

\begin{itemize}
\item[(1)] The idea of  ``a property of the system" (i.e. ``the value of a physical
quantity'') is meaningful, and representable in the theory.
\item[(2)] Propositions about the system are handled using Boolean logic. This
requirement is compelling in so far as we humans think in a Boolean way.
\item[(3)] There is a space of ``microstates'' such that specifying a microstate leads
to unequivocal truth values for all propositions about the system. The existence
of such a state space is a natural way of ensuring that the  first two requirements
are satisfied.
\end{itemize}

The standard interpretation of classical physics satisfies these requirements,
and provides the paradigmatic example of a realist philosophy in science. On
the other hand, the existence of such an interpretation in quantum theory is
foiled by the famous Kochen-Specker theorem." \cite{DoringIsham08}}
\end{quotation}

\noindent These presuppositions are, in great measure, analogous to those found within Bohr's approach to QM ---namely, the need of a classical representation to account for phenomena and the existence of a reductionistic limit between QM and classical physics. Indeed, the first point makes explicit the necessity of retaining the classical language and representation of physics, in terms of ``systems'' and ``properties'', also for QM.  The second presupposition goes in the same direction and points to the requirement that classical logic must be considered as a necessary precondition for thinking about any physical theory ---including, obviously, QM. (1) and (2) might thus seem to assume implicitly also the Bohrian claim according to which ``it would be a misconception to believe that the difficulties of the atomic theory may be evaded by eventually replacing the concepts of classical physics by new conceptual forms.'' (3) makes explicit the assumption of an atomist metaphysical picture. The reductionistic understanding of QM with respect to classical physics is also captured by the presupposition that (classical) macrosistems necessarily emerge from (quantum) microsystems (i.e., the quantum to classical limit).

As they explain in \cite[p. 1]{IshamButter1}, the topos approach is based on the introduction of ``a new type of valuation which is defined on all operators, and which respects an appropriate version of the functional composition principle. The truth-values assigned to propositions are (i) contextual; and (ii) multi-valued, where the space of contexts and the multi-valued logic for each context come naturally from the topos theory of presheaves.'' Before addressing this specific proposal for a ``proper valuation'' in QM, let us briefly introduce the mathematical scheme in which the topos approach is grounded.   

Let us start by defining the notion of a \emph{sieve}. Let $\mathcal{C}$ be a category and let $A\in\mathcal{C}$.
A \emph{sieve} on $A$ is a collections $S$ of maps over $A$ closed under precompositions. Specifically,  if the map $f:B\to A$ is in $S$ and  $g:C\to B$ is an arbitrary map, then $gf\in S$. Let us denote by $\mathcal{S}ie(A)$ to the set of sieves on $A$. It is a fact that $\mathcal{S}ie(A)$ is a Heyting algebra and that $\mathcal{S}ie:\mathcal{C}\to\mathbf{Hey}$ is a functor, where $\mathbf{Hey}$ is the category of Heyting algebras.

In order to define a sieve valued valuation as in  \cite{IshamButter1} we need to define the category $\mathcal{O}$
and the \emph{coarse-graining} presheaf $\mathcal{B}ool$.
Objects in $\mathcal{O}$ are bounded self-adjoint operators on a Hilbert space and we define an arrow $f:B\to A$  if and only if $B=f(A)$. The  \emph{coarse-graining} presheaf $\mathcal{B}ool$ is a contravariant functor from $\mathcal{O}$ to
$\mathbf{Bool}$ sending $A\in\mathcal{O}$ to the Boolean algebra generated by $\{f(A)\,\colon\,f:\sigma(A)\to\mathbb{R}\}$. 

A \emph{genaralized valuation} is a natural map
$\nu:\mathcal{B}ool\to \mathcal{S}ie$ satisfying
functional composition, null proposition condition, 
monotonicity, exclusivity and/or unit proposition condition.
Specifically, to each bounded self-adjoint operator $A$, a function
$\nu_A:\mathcal{B}ool(A)\to \mathcal{S}ie(A)$ such that
\begin{itemize}
\item Functional composition: 
$\nu_{h(A)}=h^*\nu_A$
for any Borel function $h:\sigma(A)\to\mathbb{R}$.
\item Null proposition condition: $\nu_{A}(0)=0$.
\item Monotonicity: If $P_1\le P_2$ then 
$\nu_A(P_1)\subseteq \nu_A(P_2)$.
\item Exclusivity: If $P_1\wedge P_2=0$ and $\nu_A(P_1)=1$, 
then $\nu_A(P_2)\subsetneq 1$.
\item Unit proposition condition: $\nu_{A}(1)=1$.
\end{itemize}
According to the authors, the ``size'' of the sieve $\nu_A(P)$ 
determines the degree of the partial truth of the proposition $P$. As they explain [{\it Op. cit.}, p. 19]: ``one wants the proposition A to be `more true', the greater the set of such points s."
Let us remark that, (i) the truth-value of a proposition belongs to a logical structure that is larger than $\{0, 1\}$; and (ii)  these target-logics are context-dependent.

In the second of this series of papers Isham and  Butterfield discuss the conceptual and intuitive level of analysis in order to justify the introduction of partial valuations. In order to do so they go back to classical physics arguing that: ``the notion of a classical macrostate motivates the classical analogue of a partial valuation, and thereby leads to the associated sieve-valued valuations." From this standpoint they explain the following:
\begin{quotation}
\noindent {\small ``So suppose we are given, not a microstate $s \in S$, but only a  macrostate, represented by some Borel subset $R \subseteq S$: what then can be said about  the `value' of a quantity $A$, or the truth-value of a proposition `$A \in \Delta$'. Various responses are possible: for example, the obvious choice is simply to say that the proposition $A$ is true in the macrostate $R$ if $A(R) \subseteq \Delta$, and false otherwise. Thus `$A \in \Delta$' is defined to be true if, for {\it all} microstates $s$ in $R$, the value $A(s)$ lies in the subset $\Delta$. 

However, one may feel that this assignment of true and false is rather
undiscriminating in so far as the proposition `$A \in \Delta$' is adjudged false irrespective of whether $A(s)$ fails to be in $\Delta$ for all $s \in R$, or does so only for a `few' points. For this reason, a more refined response is to say that one wants the proposition `$A \in \Delta$' to be `more true', the greater the set of such points s: an idea that can be implemented  by defining, for example, a generalised truth-value $v^R(A \in \Delta)$ of the proposition `$A \in \Delta$' to be the set of such points: $v^R(A \in \Delta)= R \cap \overline{A}^{-1}[\Delta]$"\cite[p. 20]{IshamButter2}}
\end{quotation}

As remarked above, the valuations investigated by the topos approach take for granted the orthodox reductionistic perspective according to which QM must be reduced to classical physics and classical (Boolean) logic in some ``limit''. And this seems to be the reason why the notions of ``coarse graining'' and ``partial truth'' become so important within this approach. We will come back to a more detailed analysis of these important points in section 5. In the next section we will present the logos categorical approach to QM which attempts to provide a different answer to the problem and understanding of quantum contextuality. Contrary to the topos approach, our line of research attempts to retain the main features of the orthodox quantum formalism and provide a non-reductionistic conceptual representation of them.

\section{The Logos Categorical Approach to QM}

In this section we present a mathematical formalism based on category theory that will allow us to introduce intensive valuations (section 6). Let us start with some constructions and definitions. Our main reference is \cite{maclane98}. First of all, a \emph{category} consists of a collection of \emph{objects} (often denoted as $X,Y,A,B$), a collection of \emph{morphisms} (or \emph{arrows}, denoted $f,g,p,q$) and four operations,
\begin{itemize}
\item To each arrow $f$, there exists an object $dom(f)$, called its \emph{domain}.
\item To each arrow $f$, there exists an object $codom(f)$, called its \emph{codomain}.
\item To each object $X$, there exists an arrow $1_X$, called the \emph{identity map} of $X$.
\item To each pair of arrows $f,g$ such that $dom(g)=codom(g)$ there exists a \emph{composition map}, $fg$ such that
$dom(fg)=dom(g)$ and $codom(fg)=codom(f)$.
\end{itemize}
An arrow $f$ is often denoted as $f:X\rightarrow Y$ to empathizes
the fact that $dom(f)=X$ and $codom(f)=Y$. 
We say that an arrow $f:X\rightarrow Y$ is invertible if there
exists an arrow $g:Y\rightarrow X$ such that $fg=1_Y$ and $gf=1_X$.
The collection of arrows between $X$ and $Y$ is denoted $\hom(X,Y)$. The collection of categories has a structure of a category. The arrows are called \emph{functors}. 
A functor $F:\mathcal{C}\rightarrow \mathcal{D}$ assign objects to
objects, arrows to arrows and is compatible with the
four operations (domain, codomain, identity and composition).

Let us present three standard constructions in category theory,
the \emph{comma category}, the \emph{graph of a functor}
and the \emph{category over an object}. The second construction
is a particular case of the first and the third of the second. Let $F:\mathcal{A}\rightarrow \mathcal{C}$ and  $G:\mathcal{B}\rightarrow \mathcal{C}$ be two functors with
the same codomain. 
The \emph{comma category} $F|G$
is a category whose objects are arrows in $\mathcal{C}$ of
the form 
\[
f:F(A)\rightarrow G(B),
\]
where $A\in\mathcal{A}$ and $B\in\mathcal{B}$. 
An arrow between $f$ and $g$ is a commutative square.

The \emph{graph of a functor} $F:\mathcal{A}\rightarrow \mathcal{C}$
is defined as the comma categoty $F|1$,
where 
$1=1_\mathcal{C}:\mathcal{C}\rightarrow\mathcal{C}$ 
is the identity functor. 

In the special case where
the functor $F$ is equal to $\hom(-,C)$ for some object 
$C\in\mathcal{C}$, the graph of this functor is called 
the category over $C$ and is denoted $\mathcal{C}|_C$.
This is our main structure. Objects in $\mathcal{C}|_C$
are given by arrows to $C$, $p:X\rightarrow C$, 
$q:Y\rightarrow C$, etc. Arrows $f:p\rightarrow q$
are commutative triangles,
\[
\xymatrix{
X\ar[rr]^f\ar[dr]_p& &Y\ar[dl]^q\\
&C
}
\]

\begin{example}
Let $\mathcal{S}ets|_\mathbf{2}$ be the category of sets
over $\mathbf{2}$, where $\textbf{2}=\{0,1\}$ and $\mathcal{S}ets$ is
the category of sets.
Objects in $\mathcal{S}ets|_\mathbf{2}$
are functions from a set to $\{0,1\}$ (that is, GBV)
and morphisms are commuting triangles, 
\[
\xymatrix{
\mathcal{G}_1\ar[rr]^f\ar[dr]_\Psi& &\mathcal{G}_2\ar[dl]^\Phi\\
&\{0,1\}
}
\]
In the previous triangle, $\Psi$ and $\Phi$ are objects of 
$\mathcal{S}ets|_\mathbf{2}$
and $f$ is a function satisfying $\Phi f=\Psi$.

As we mention before, this category is relevant in classical logic.
We can assign a true/false value to every element of $\mathcal{G}_1$
and/or $\mathcal{G}_2$ in a consistent manner. We say that $P\in \mathcal{G}_1$ 
(assume $P$ is a \emph{proposition} in the \emph{universe} $\mathcal{G}_1$) is true
if $\Psi(P)=1$, else we say that $P$ is false. 
Even more so, assume that we have a map $f:\mathcal{G}_1\rightarrow \mathcal{G}_2$
such that $\Phi f=\Psi$, then the truth or falsity of $P$ is unchanged
via $f$, that is $f(P)$ is true if and only if $P$ is true,
\[
f(P)\mbox{ is true }
\Longleftrightarrow 1=\Phi (f(P))=\Psi(P)
\Longleftrightarrow 
P\mbox{ is true }
\]
\end{example}

We are going to generalize the category $\mathcal{S}ets|_\mathbf{2}$
to allow a mathematical formalism to work with distinct logics.
One of the first ideas in this direction came from Weyl's work in 1949 on \emph{aggregates}, i.e. a
set with an equivalence relation,  \cite[App.B]{weyl49}. 
We generalize the definition of aggregates
by considering sets with a reflexive, symmetric relation.
Such an object is called
a \emph{simple graph} or just a \emph{graph}.
We denote by $\mathcal{G}ph$ to the 
the category of graphs. 
It extends naturally, the category of sets and the category
of aggregates.
This category has very nice categorical properties, 
\cite{quasitopoi, graphtheory}.

\begin{example}
A set is a graph without edges. An aggregate is a graph, 
in which, the relation is transitive. More generally, 
we can assign to a category a graph, where the objects
are the nodes of the graph and there is an edge
between $A$ and $B$ if $\hom(A,B)\neq\emptyset$.
Given that in a category we have a composition law, 
the resulting graph is an aggregate.

A more interesting case is that we can view 
an orthomodular lattice $\mathcal{L}$ as a graph.
We say that $x\in\mathcal{L}$ commutes with $y\in\mathcal{L}$, denoted $xC y$,
if $x=(x\wedge y)\vee (x\wedge y^\bot)$.
Clearly, $C$ is symmetric.
It is known that $C$ is reflexive if and only if $\mathcal{L}$
is orthomodular. Also $C$ is an equivalence relation 
if and only if $\mathcal{L}$ is Boolean.
Then, any orthomodular lattice with the commutativity
relation is a graph.
\end{example}

\begin{definition}
We say that a graph $\mathcal{G}$ is \emph{complete}
if there is an edge between two arbitrary nodes.
\end{definition}

An object in $\mathcal{G}ph|_{[0,1]}$ consists in 
a map $\Psi:\mathcal{G}\rightarrow [0,1]$, where $\mathcal{G}$ is a graph. 
We say that $\Psi$ is a \emph{Global Intensive Valuation}. Intuitively, $\Psi$
assigns a \emph{potentia} to each node of the graph $\mathcal{G}$.
Specifically, to each node 
$P\in\mathcal{G}$, we assign a number $\Psi(P)$, but this time, 
$\Psi(P)$ is a number between $0$ and $1$.

\begin{example}
Let $\mathcal{H}$ be Hilbert space and let $\Psi$ be a vector, $\|\Psi\|=1$.
Take $\mathcal{G}$ as the set of observables with the commuting relation
given by QM. This relation is \textbf{not} transitive, hence
$\mathcal{G}$ is a non-complete graph. 
To each observable $P\in\mathcal{G}$
apply the Born rule to get the number $\Psi(P)\in[0,1]$. 
Then, $\Psi:\mathcal{G}\rightarrow [0,1]$
defines an object in $\mathcal{G}ph|_{[0,1]}$.
We call this map a \emph{Potential State of Affair}.
\end{example}

The following definitions were taken from \cite{freyd}.
Let $\mathbf{Pos}$, $\mathbf{Lat}$ and $\mathbf{Hey}$
be the catergories of posets, lattices and Heyting alegebras.
Recall that a Heyting algebra is a lattice with an implication operation.
Let $\mathcal{C}$ be an arbitrary category with pullbacks.
A \emph{subobject} of an object $X$ in $\mathcal{C}$ 
is an isomorphism class of a monomorphism $i:S\to X$.
Two morphisms $i:S\to X$ and $j:T \to X$ are isomorphic if there exists an isomorphism $k:S\to T$ such that $i=jk$. The set $\mbox{Sub}(X)$ of subobjects of $X$ has an order. 
We say that $i\leq j$
if there exists $k$ such that $i=jk$,
\[
\xymatrix{
S\ar@{-->}[d]_{\exists k}\ar[r]^i\ar@{}[dr]|<<<<\equiv&X\\
T\ar[ur]_j&
}
\]
The order given in $\mbox{Sub}(X)$ is functorial.
If $f:X\to Y$, then
$f^*:\mbox{Sub}(Y)\to\mbox{Sub}(X)$ is 
an order-preserving function, \cite[1.451]{freyd}.
We can rephrase this by saying that 
$\mbox{Sub}:\mathcal{C}\to\mathbf{Pos}$ is a functor.
A category $\mathcal{C}$ is called a \emph{pre-logos}
if $\mbox{Sub}:\mathcal{C}\to\mathbf{Lat}$,
and is called a \emph{logos}
if $\mbox{Sub}:\mathcal{C}\to\mathbf{Hey}$.

Now that we know what a logos is, let us conclude by saying that $\mathcal{G}ph$ is a logos. It follows
from the fact that subgraphs of a graph form a Heyting algebra, \cite[3.3]{Angot15}. Let us also remark that it is possible to give to a logos a logic.  The logic defined by $\mathcal{G}ph$  is intuitionist (no binary valuations) and  paraconsistent (inconsistency-tolerant). In the next section we will turn our attention to the way KS type contextuality can be visualized in terms of graphs.

\section{KS Contextuality and  Graphs}  

KS theorem has been criticized in the literature for being ``too complicated'' and ``too abstract''.\footnote{For example Mermin states in \cite[p. 3375]{Mermin90a} that: ``The strengthening of Bell's example by GHZ and the demonstration that GHZ also works as a KS theorem should liberate future generations of students of the foundations of quantum mechanics from having to cope with the geometrical intricacies of the original KS argument.''} Indeed, the original mathematical derivation in \cite{KS} with 117 rays was quite difficult to follow for the average physicist. However, a few decades after, the work by Asher Peres and David Mermin ---between many others--- reduced the number of rays that had to be considered in the proof and developed more visualizable representations which made possible to better grasp the meaning of the theorem. In this respect, even though category theory is a rigorous and theoretically rich platform, it can also provide a visualizable approach to the seemingly complicated problem implied by KS contextuality through the use of graphs. In \cite{cabello} such a proof was developed. The authors constructed 18 vectors $\{v_1,\ldots,v_{18}\}$ in a 4-dimensional vector space such that there are 9 orthogonal bases and every vector is in exactly two of these bases. KS Theorem says that it is not possible to assign a $\{0,1\}$ value to every vector in such a way that $\sum_{j=1}^4 \nu(v_{i_j})=1$,  for every orthogonal basis $\{v_{i_1},v_{i_2},v_{i_3},v_{i_4}\}$. The proof is very simple and visualizable. By adding the 9 equations,  $\sum_{j=1}^4 \nu(v_{i_j})=1$, we get an even number in the right hand side (due to the repetitions of the vectors in each basis) and an odd number (nine) on the left. A contradiction. Hence, no such binary valuation exists.

We can represent the vectors in the following graph:
\begin{center}
\includegraphics[width=16em]{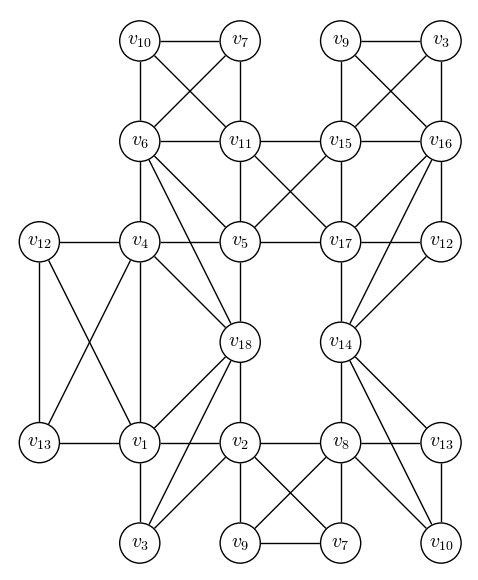}
\end{center}
{\footnotesize
\[
v_{1}=(0,0,1,0),\,
v_{2}=(1,1,0,0),\,
v_{3}=(1,-1,0,0),\,
v_{4}=(0,1,0,0),\,
v_{5}=(1,0,1,0),\,
v_{6}=(1,0,-1,0),\,
\]
\[
v_{7}=(1,-1,1,-1),\,
v_{8}=(1,-1,-1,1),\,
v_{9}=(0,0,1,1),\,
v_{10}=(1,1,1,1),\,
v_{11}=(0,1,0,-1),\,
v_{12}=(1,0,0,1),\,
\]
\[
v_{13}=(1,0,0,-1),\,
v_{14}=(0,1,-1,0),\,
v_{15}=(1,1,-1,1),\,
v_{16}=(1,1,1,-1),\,
v_{17}=(-1,1,1,1),\,
v_{18}=(0,0,0,1).
\]
}

\noindent Each node is a vector $v_i$ and each edge represents
the orthogonality relation; 
if there exists an edge between $v_i$
and $v_j$ then $v_i\bot v_j$, or equivalently, 
$|v_i\rangle\langle v_i|$ commutes with 
$|v_j\rangle\langle v_j|$. 
We only draw the edges relevant to
the 9 orthogonal 
basis (the contexts) leaving aside the relations,
\[
v_4\bot v_9,\,v_6\bot v_{16},\,v_1\bot,v_{11},\,
v_2\bot v_{17},\,v_8\bot v_5,\,v_{14}\bot v_{18},\,v_{3}\bot v_{10},\,v_{7}\bot v_{12},\,v_{13}\bot v_{15}.
\]
Now, in order to visualize the 9 different contexts involved within our graph we just need to see the nodes and circle them. 
\begin{center}
\includegraphics[width=16em]{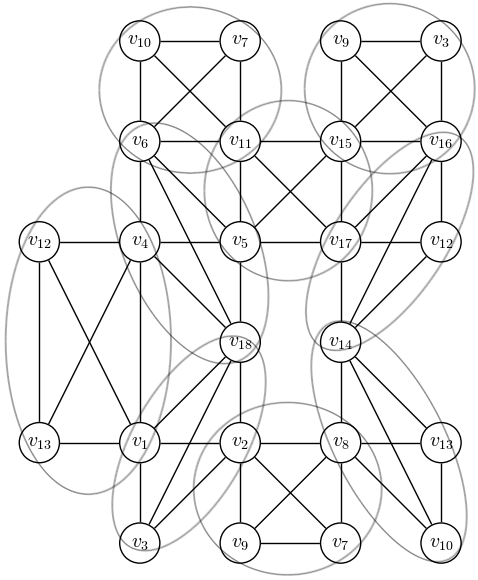}
\end{center}
A closer inspection on the nodes, shows that
there are repeated nodes, 
\[v_{12}, v_{13}, v_{3}, v_{9}, v_{7}, v_{10}.\]
This is done to avoid double crossing of edges. It is a nice exercise to try to test KS Theorem
by filling the circles in red or blue ($0$ or $1$) in such a way that the nodes in each context have 
the same color. This shows the power of our representation.

As we discussed above the impossibility to provide consistently values, $0$ or $1$, to the elements of the graph presents a difficulty for the objective account of the model interpreted in terms of properties with definite preexistent values. In the next two sections we will discuss, firstly, the coarse-graining contextual proposal investigated by the topos approach in terms of {\it sieve-valued valuations}, and secondly, our own logos proposal which taking as a standpoint the physical invariance of the Born rule introduces the notion of {\it intensive valuation} in order to restore the possibility of an objective (non-contextual) account of projection operators.

\section{The Topos Approach: Contextual Sieve-Valued Valuations}

KS theorem proofs there does not exist a GBV for the projection operators of a quantum state. This is the reason why Isham and Butterfield developed a generalization of the notion of valuation in terms of what they called {\it sieve-valued valuation}. As they explain:
\begin{quotation}
\noindent {\small ``In the [topos] programe to be discussed here, 
ual valuation will be developed in a different direction from that of the existing modal interpretations. In particular, rather than accepting only a limited domain of beables we shall propose a theory of `generalized' valuations that are defined globally on all propositions about values of physical quantities. However, the price of global existence is that any given proposition may have only a `partial' truth-value. More precisely, (i) the truth-value of a proposition `$A \in \Delta$' belongs to a logical structure that is larger than $ \{0, 1 \}$; and (ii) these target-logics are context-dependent." \cite[p. 5]{IshamButter1}}
\end{quotation}

\noindent It is important to stress that their approach also attempts to restore distributive logic for each chosen context. This goes in line with the widespread argument that once the context has been chosen in QM classicality is automatically restored: 
\begin{quotation}
\noindent {\small ``It is clear that the main task is to formulate mathematically the idea of a contextual, `partial' truth-value in such a way that the assignment of generalized truth-values is consistent with an appropriate analogue of the functional composition principle FUNC. The scheme also has to have some meaningful physical interpretation; in particular, we want the set of all possible partial truth-values for any given context to form some sort of distributive logic, in order to facilitate a proper semantics for this `neo-realist' view of quantum theory." \cite[p. 5]{IshamButter1}}
\end{quotation}

As remarked above, there are two important aspects of the topos approach which go in line with the Bohrian program that attempted to understand QM as a rational generalization of classical physics \cite{BokulichBokulich}. On the one hand, they both share an emphasis on the classical representation. For example, D\"oring and Isham \cite{DoringIsham11} argue that: ``we and our collaborators have shown how quantum theory can be re-expressed as a type of `classical physics' in the topos of  presheaves (i.e., set-valued contravariant functors) on the partially-ordered set all commutative von Neumann sub-algebras of the algebra of all bounded operators on the quantum-theory Hilbert space ${\cal H}$.'' It becomes clear that their goal is to restore a classical account even at the price of leaving aside some of the main features present within the quantum formalism. D\"oring and Barbosa make this point explicit: 
\begin{quotation}
\noindent {\small ``The so-called topos approach provides a radical reformulation of quantum theory. Structurally, quantum theory in the topos formulation is very similar to classical physics. There is a state object $\Sigma$, analogous to the state space of a classical system, and a quantity-value object $\mathbb{R}^{\leftrightarrow}$, generalising the real numbers. Physical quantities are maps from the state object to the quantity-value object $\mathbb{R}^{\leftrightarrow}$ hence the `values' of physical quantities are not just real numbers in this formalism. Rather, they are families of real intervals, interpreted as `unsharp values'.'' \cite[p. 1]{DoringBarbosa}}
\end{quotation}

But, as noticed above, the analogy to Bohr's approach does not only regard the exclusive use of classical notions in order to restore a classical understanding of what QM is talking about. It also relates to the solution provided to KS-contextuality through the {\it principle of complementarity}. As Benjamin Eva \cite{Eva} explains:
\begin{quotation}
\noindent {\small ``We have seen that in TQT [Topos Quantum Theory], the quantum state space is formalised as a kind of amalgamation of the local, classical state spaces of each of the `classical perspectives' on V (also referred to as `contexts'). `The topos approach emphasises the role of classical perspectives onto a quantum system. [...] One of the main ideas is that all classical perspectives should be taken into account simultaneously. (D\"oring, 2011) This emphasis on classical perspectives is deeply reminiscent of Bohr's famous `principle of complementarity'(PC), which is neatly summarised by the claim that `Talk of the position of an electron has sense only in the context of an experimental
arrangement for making a position measurement.' (Bohr, quoted in Gibbins [1987]) The philosophical upshot of PC is that physical propositions about a quantum system can only be made with reference to some fixed classical perspective on that system. This notion is taken seriously in TQT, and is evident in the way that physical propositions are eventually formalised.'' \cite{Eva}}
\end{quotation}

The topos approach might be regarded as one of the first to inaugurate the research of QM using category theory (see also \cite{Abramsky, CoeckeKissinger17}). This categorical research related to QM has become a field growing rapidly within the literature. However, the proposal of Isham and Butterfield has several drawbacks and difficulties which we will now discuss and analyze in some detail (see for a more detailed analysis \cite{Eva}). 

Firstly, the topos proposal does not evade KS contextuality completely. The term ``global'' is used in the topos approach in order to refer to the ``global sections of sheaves'', not to {\it global valuations} ---which is what KS theorem discusses about. In this respect, Isham and Butterfield remark themselves that ``the truth-values assigned to propositions [in the topos approach] are contextual''.  Explained in formal terms, they define a generalized valuation as a global section $\nu\in \mathbf{\Omega}^{\mathbf{G}}$. Then for every bounded self-adjoint operator $A$ on some Hilbert space, there is a map $\nu_A:\mathbf{G}(A)\to\mathbf{\Omega}(A)$ sending elements $P$ in the spectral algebra  of $A$ to a sieve on $A$, $\nu_A(P)$. Clearly,  a \emph{global} true/false assignment to elements in  $\mathbf{G}(A)$ for all $A$ is precluded by KS Theorem. Recall that $\mathbf{\Omega}(A)$ is a Heyting algebra and contains $\{0,1\}$.\footnote{The idea of a contextual or partial truth valuation has been recently formalized by Karakostas and Zafiris \cite{Karakostas14, KarakostasZafiris16}. Our approach goes in line with the development of a {\it potential truth} as discussed in \cite{deRonde17b}. The analysis of this specific subject exceeds the scope of this paper and will be left for a future work.} Unfortunately, this extension of the notion of valuation in terms  of sieve valued valuation does not seem to provide a clear physical explanation nor intuition of what is really going on. 

Secondly, we should remark that {\it binary valuations} continue to play a central role within the topos approach. The notions of ``partial truth'' and ``coarse graining''  introduce within the approach, in close analogy to classical physics, an epistemic realm of analysis. However, this analysis rests, like in the case of classical physics, in the implicit presupposition that properties do posses in fact definite truth values. Take for example a table as described by classical mechanics in terms of a rigid body. Obviously, this physical representation presupposes that the table (i.e., the rigid body) has a definite length $L$. The fact that we will never be able to measure it exactly is due to the limit of accuracy,  $\Delta l$, given by the precision of any apparatus used to measure the system. This is of course an epistemic problem, not an ontological one. Now, suppose that `$A=a_1$' and `$A=a_2$' are two propositions, $a_1\neq a_2$. The idea behind their generalized valuations is to find a \emph{partial truth} to a proposition like `$A=a$' by finding a \emph{coarse-grained} operator $f(A)$ such that the weaker proposition `$f(A)=f(a)$' is totally true, see \cite[p.9]{IshamButter1}. As they explain:
\begin{quotation}
\noindent {\small ``Physically, the inequality in Eq. (1.8) reflects the fact that the proposition `$f(A) \in f(\Delta)$' is generally weaker than the proposition `$A\in\Delta$' in the sense that the latter implies the former, but not necessarily vice versa. For example, the proposition `$f(A) = f(a)$' is weaker than the original proposition `$A = a$' if the function $f$ is many-to-one and such that more than one eigenvalue of  $\hat{A}$ is mapped to the same eigenvalue of $f(\hat{A})$. In general, we shall say that `$f(A)\in f(\Delta)$' is a \emph{coarse-graining} of `$A\in\Delta$'.'' \cite[p. 6]{IshamButter1}
}\end{quotation}

\noindent Still, the point which remains unclear from a physical perspective is the meaning of a quantum property, $A$, which does not possess a definite value but rather pertains to an interval, $A\in\Delta$. Clearly, even though a classical object such as a table possesses the property of having a definite length $L$, the knowledge we can acquire through measurements is related to a value $L'\in L \pm \Delta l$. Since KS theorem imposes a limit to the value of $A$ in ontological terms, it is still not clear what is the meaning of  a value which is not definite, $A\in\Delta$?

Thirdly, it is important to notice that the topos approach does not use the same topos which expresses KS theorem. Instead of remaining within this topos, Isham and Butterfied choose to discuss a different topos which is not directly related to KS contextuality.
\begin{quotation}
\noindent {\small ``5. Generalized Valuations as Global Sections of a Presheaf: We note in passing that there is a bijection between morphisms from {\bf G} to ${\bf \Omega}$ , and global elements of the `exponential' object  ${\bf \Omega}^{\bf G}$ which, roughly speaking, is the topos analogue of the set $Y^X$ of all maps from $X$ to $Y$ in normal set theory. Thus a generalized valuation does turn out to be a global section of a certain presheaf on $\mathcal{O}$, but it is the presheaf  ${\bf \Omega}^{\bf G}$, not the simple dual presheaf 
${\bf D}\circ{\bf W}$ to which the Kochen-Specker `no-go' theorem applies.'' \cite[p. 32]{IshamButter1}}
\end{quotation}

\noindent By definition, the \emph{dual presheaf} on $\mathcal{O}$ is a contravariant functor ${\bf D}\circ{\bf W}:\mathcal{O}\to\mathcal{S}et$ such that 
\[
{\bf D}\circ{\bf W}(-)=
\mbox{Hom}(\mathbf{G}(-),\{0,1\}).
\]
Saying it differently, an element in ${\bf D}\circ{\bf W}(A)$ 
assigns a true/false value to every element 
in the spectral algebra of $A$. The difference between 
${\bf \Omega}^{\bf G}$ is that a valuation $\nu_A$
assigns values in a Heyting algebra (different from $\{0,1\}$) 
to  elements in the spectral algebra of $A$.

Last but not least, maybe the deepest drawback of the topos approach is that the highly abstract categorical tools and models used within the approach in order to account for QM do not provide an intuitive understanding of the mathematical formalism. On the contrary, the ``bridge'' designed in order to explain how QM can be understood in terms of classical physics through sieve valued valuations, coarse graining processes and partial truth is very complicated and difficult to follow even for experts within the field. According to our view, the ``conceptual explanation'' of the approach presented in the second paper of the first series, \cite{IshamButter2}, fails to provide a clear physical intuition of what is really going on according to QM. Rather than providing examples which could be followed by the working physicist, the debate becomes even more abstract and complicated than in the first paper. The simple graphs that we used in the previous section in order to visualize KS theorem cannot be worked out in the topos approach since their new formalism lyes outside of the orthodox KS framework. 

As we shall see in the next section, our logos approach to QM takes a completely different standpoint and direction than the one assumed by the topos approach. Rather than attempting to ``bridge the gap'' between the quantum formalism and classical physics, our approach takes as standpoints, firstly, the orthodox formalism of QM, and secondly, the original physical meaning of ``global'' as discussed in the KS theorem (i.e., the possibility of representing consistently through a mathematical formalism a given state of affairs). The direction of our line of research is determined by the physical principles of QM themselves. Thus, it is the principles themselves which must provide the constraints that will allow us to develop a new set of concepts which, in turn, are capable of explaining what the theory is really talking about.

\section{The Logos Approach: Global Intensive Valuations}

It is important to remark that we do not attempt to provide an epistemological analysis of our approach in this paper. Since we believe it is only the theory which can tell us what can be observed, in order to provide an epistemological analysis we must previously understand what type of experience is involved within the theory. In this respect, the logos approach does not attempt to understand QM in reductionistic terms with respect to classical physics and logic. Our goal is not to ``bridge the gap'' between the quantum formalism and our ``common sense'' classical understanding of the world. Our main goal is to develop an objective conceptual representation which allows us to understand what the theory is talking about beyond mathematical structures and measurement outcomes. Taking a pluralist perspective as a standpoint \cite{deRonde16b}, we believe that QM can be developed by considering new (non-classical) physical concepts that, in turn, will provide us with an intuitive (non-classical) representation and understanding of the theory and the experience it talks about. Our guide in the development of such new conceptual forms is the orthodox quantum formalism itself assisted by some general physical considerations. Let us begin by the latter.  

The notion of invariance has always played a major role within the development of theories. As remarked by Max Born \cite{Born53}: ``the idea of invariant is the clue to a rational concept of reality, not only in physics but in every aspect of the world.'' The notion of invariance allows us to determine what is to be considered {\it the same} within a mathematical formalism regardless of any particular reference frame. In physics, invariants are quantities which can be translated ---through mathematical transformations--- from one reference frame to another. Even though the values of physical magnitudes might vary from one frame to another ---due to the dynamics between reference frames---, in both classical physics and relativity theory there is a {\it consistent translation} between the values of magnitudes of different frames secured by the specific laws of transformation ---Galilean transformations in the case of classical physics and Lorentzian in the case of relativity. As a consequence, within these theories it is always possible to provide a GBV of the properties that pertain to systems. For example, in classical physics, the position of a rabbit running through the fields and observed by a distant passenger of a high speed train can be translated to the position of that same rabbit taken from the perspective of another passenger waiting in the platform of the station. The fact that the values of observables (position, momentum, etc.) can be consistently translated from one reference frame to the other allows us to assume that such physical observables also bear an objective real existence completely independent of the specific choice of the reference frame of the observers. The rabbit has a set of dynamical properties (position, a momentum, etc.) independently of his observers in the train and on the platform. This definition of invariant observables allows us to argue that the physical system (the rabbit) possesses preexistent properties independently of the observers, their choice of the reference frame or even the measurements they might perform. Consequently we are allowed to claim that such properties are {\it dynamical variations} that pertain to {\it the same} physical system. It is this type of analysis which allows us to understand what can be considered as an {\it objective representation} (see for a detailed analysis \cite{deRonde17b, deRondeMassri16}). Invariance is our thread of Ariadne in the quantum labyrinth of conceptual representation. But since this quantum labyrinth is clearly not classical, we should not expect classical concepts to help us finding the way out. And indeed, for almost a century now, no researcher entering this labyrinth has been able to escape with classical ideas. It seems we might need to think differently. Maybe we should leave aside our classical metaphysical prejudices and commitments and simply follow the thread. 

Taking into account our analysis about the importance of invariance within physical theories, let us consider the following simple question within QM: given a $\Psi$, what is independent of the relative observational choice of individual subjects according to the formalism? Or in other words, what is invariant with respect to mathematical contexts (or bases)? The answer for any quantum physicist who knows the orthodox formalism is obviously straightforward. The invariants in QM are the average value of observables considered in relation to that particular vector. The average value of an observable is independent of the particular context (or basis) that one might chose in order to calculate the computation. This is known by quantum physicists as the Born rule. \\

\noindent {\it
{\bf Born Rule:} Given a vector $\Psi$ in a Hilbert space, the following rule allows us to predict the average value of (any) observable $P$. 
$$\langle \Psi| P | \Psi \rangle = \langle P \rangle$$
This prediction is independent of the choice of any particular basis.}\\

From the previous considerations, we believe that a good starting point in order to derive an objective representation for QM must rely on the invariants of the formalism. Also, recalling Einstein's remark to Heisenberg \cite[p. 63]{Heis71} that: ``It is only the theory which decides what can be observed'', we understand it is the Born rule which provides the empirical constraints to consider what can be observed according to QM. Following this line of thought it makes then perfect sense to take Born's rule ---which provides not only the invariant elements of the formalism but also the constraints to observation--- very  seriously. This means for us to consider on equal footing ---irrespectively of any actualist metaphysical commitment--- both {\it certain predictions} (probability equal to unity) and {\it statistical predictions} (probability between zero and unity). Of course this implies the abandonment of the classical understanding of probability in terms of `ignorance about an actual state of affairs' and develop instead a new understanding of probability in terms of objective knowledge. But how to do so in relation to physical reality? We have argued elsewhere \cite{deRonde16a}, that this move requires the generalization of the famous definition of an {\it element of physical reality} presented in the famous EPR paper \cite{EPR}.\\

\noindent {\bf Element of Physical Reality:} {\it If, without in any way disturbing a system, we can predict with certainty (i.e., with probability equal to unity) the value of a physical quantity, then there exists an element of reality corresponding to that quantity.}\\

\noindent As remarked by Aerts and Sassoli de Bianchi \cite[p. 20]{AertsSassoli}: ``the notion of `element of reality' is exactly what was meant by Einstein, Podolsky and Rosen, in their famous 1935 article. An element of reality is a state of prediction: a property of an entity that we know is actual, in the sense that, should we decide to observe it (i.e., to test its actuality), the outcome of the observation would be certainly successful.'' Indeed, certainty and actuality where the restrictive constraints of what could be considered as physically real. However, there is a different path that can be considered, namely, to redefine the meaning of valuation itself beyond the actual realm through an ontological ---rather than epistemic--- generalization. This redefinition will also imply, as a direct consequence, the reconsideration of the meaning of quantum physical reality beyond actuality. Of course, our redefinition should keep the necessary relation between operational predictive statements and physical reality. But in this case, leaving aside both the actualist constraint imposed by certainty (probability equal to unity) and the (subjective) intromission of measurement (or observation) within the physical representation of the state of affairs. Following these considerations and constraints we can now put forward a generalized idea of what should be considered to be an element of physical reality in the specific context of QM \cite{deRonde16a}.\\

\noindent {\it {\bf Generalized Element of Physical Reality:} If we can predict in any way (i.e., both probabilistically or with certainty) the value of a physical quantity, then there exists an element of reality corresponding to that quantity.}\\

\noindent By extending the limits of what can be considered as physically real, we have also opened the door to a new understanding of QM beyond the representation provided by classical metaphysics in terms of systems composed by definite valued properties (see also \cite{deRonde17b}). 

Taking the Born rule as a standpoint, we will now consider a generalized notion of valuation which ---in line with our previous definition--- goes beyond the restrictive binary valuation imposed by actualist metaphysics and is grounded on the formalism of QM itself.\\

\noindent{\it
{\bf Global Intensive Valuation:} 
A Global Intensive Valuation (GIV) is a function from a graph to 
the closed interval $[0,1]$, that is, a GIV is an object in $\mathcal{G}ph|_{[0,1]}$.
}\\

\noindent Taking into account our intensive valuation, we can now proceed to discuss a new representation of physical reality beyond the notion of {\it Actual State of Affairs} which relates the set of properties $\mathcal{G}$ with the truth values $\{0,1\}$ according to the function $\Psi:\mathcal{G}\rightarrow\{0,1\}$ (section 1). We shall call this new conceptual representation of (quantum) physical reality a \emph{Potential State of Affairs} (PSA). 
To give the definition of a PSA, we need to introduce the graph of observables. Let $\mathcal{H}$ be a Hilbert space and let $\mathcal{G}=\mathcal{G}(\mathcal{H})$  be the set of observables. We give to $\mathcal{G}$ a graph structure by assigning an edge between
observables $P$ and $Q$ if and only if $[P,Q]=0$. Among all global intensive valuations we are interested in the particular class of PSA.
\begin{definition}
Let $\mathcal{H}$ be a Hilbert space.
A \emph{Potential State of Affairs}\footnote{A similar definition is discussed in \cite{Kalmbach}.} is a global intensive valuation
$\Psi:\mathcal{G}(\mathcal{H})\to[0,1]$ from the graph of observables $\mathcal{G}(\mathcal{H})$
such that $\Psi(I)=1$ and 
\[
\Psi(\sum_{i=1}^{\infty} P_i)=
\sum_{i=1}^\infty \Psi(P_i)\]
for any piecewise orthogonal projections $\{P_i\}_{i=1}^{\infty}$.
The numbers $\Psi(P) \in [0,1]$, are called {\it intensities} or {\it potentia}
and the nodes $P$ are called \emph{immanent powers} (or \emph{power}).
Hence, a PSA assigns a potentia to each power.
Notice that the definition of PSA is non-contextual. 
\end{definition}
Intuitively, we can picture a PSA
as a table,
\[
\Psi:\mathcal{G}(\mathcal{H})\rightarrow[0,1],\quad
\Psi:
\left\{
\begin{array}{rcl}
P_1 &\rightarrow &p_1\\
P_2 &\rightarrow &p_2\\
P_3 &\rightarrow &p_3\\
  &\vdots&
\end{array}
\right.
\]

\begin{theo}\label{teo1}
Let $\mathcal{H}$ be a separable Hilbert space, $\dim(\mathcal{H})>2$ and let $\mathcal{G}$ be the graph of immanent powers with the commuting relation given by QM.
\begin{enumerate}
\item Any positive semi-definite self-adjoint operator 
of the trace class $\rho$ determines in a bijective way
a PSA $\Psi:\mathcal{G}\to [0,1]$. 
\item Any GIV determines univocally a GBV such that the set of powers are considered as potentially existent. 
\end{enumerate}
\end{theo}

\begin{proof}
\begin{enumerate}
\item Using Born's rule, we can assign to each
observable $P\in\mathcal{G}$ the value $\mbox{Tr}(\rho P)\in[0,1]$.
Hence, we get a PSA $\Psi:\mathcal{G}\to[0,1]$.
Let us prove that this assignment is bijective. Let 
$\Psi:\mathcal{G}\to[0,1]$ be a PSA. By Gleason's theorem \cite{Gleason}
there exists a unique positive semi-definite self-adjoint operator 
of the trace class $\rho$
such that $\Psi$ is given by the Born rule with
respect to $\rho$.
\item Consider the function $\tau:[0,1]\to\{0,1\}$, 
where $\tau(t)=0$ if and only if $t=0$. Now, given a 
GIV $\Psi:\mathcal{G}\to[0,1]$, the map $\tau \Psi:\mathcal{G}\to\{0,1\}$
is a well-defined GBV. 
\end{enumerate}\qed
\end{proof}

\smallskip

The importance of Item 1 is the equivalence between a PSA and a \emph{ray} in Hilbert space (of dimension $>2$).\footnote{This is what is called following the ``minimal interpretation'' a {\it quantum state}. For a detailed analysis of the notion of quantum state in the topos approach see \cite{Eva}.} This shows that our logos approach captures completely the formal structure of the quantum formalism. Notice also that the map of Item 2 in Theorem \ref{teo1} is never, due to KS theorem, an ASA. The previous theorem also makes explicit the fact that we are loosing a lot of information when we impose binary valuations to the quantum formalism. Binary valuations contain much less information than intensive valuations. As shown by Theorem \ref{teo1} while a PSA implies univocally a GBV of potentially existent powers; the inverse doesn't hold, from a GBV we cannot derive a PSA.
$$PSA \Rightarrow GBV$$
$$GBV  \nRightarrow PSA$$

In analogous fashion to the classical case in which an ASA evolves in time, the PSA, as defined above, is perfectly well defined in terms of its observables and their respective intensities. The formalism also provides the evolution through Schr\"odinger's equation of motion.     
\[
\xymatrix{
\mathcal{G}_{t_1}\ar[rr]^{Sch}\ar[dr]_{\Psi_{t_1}}& &\mathcal{G}_{t_2}\ar[dl]^{\Psi_{t_2}}\\
&[0,1]
}
\]
\noindent Notice that in each instant of time $\mathcal{G}_{t_i}$ is defined in a univocal and non-contextual manner. 
In fact, the diagram above is capable of producing a picture different to both \emph{Schr\"odinger picture} (in which the state is fixed and the observables evolve) and to {\it Heisenberg picture} (in which the observables are fixed and the state evolves). Our diagram allows to consider a third possibility, namely, a picture in which {\it everything is evolving}. We leave a more detailed analysis of the logos picture for a future work \cite{deRondeMassri17b}.

Before discussing how our logos approach provides a new understanding of KS theorem bypassing the contextual character of {\it binary valuations}, let us provide a more detailed explanation of our {\it intensive valuations} in terms of graphs. 
\begin{definition}
Let $\mathcal{G}$ be a graph. We define a \emph{context} as a complete subgraph (or aggregate) inside $\mathcal{G}$. For example, let $P_1,P_2$ be two elements of $\mathcal{G}$. Then, 
$\{P_1, P_2\}$ is a contexts if $P_1$ is related to $P_2$, $P_1\sim P_2$. Saying it differently, if there exists an edge between $P_1$ and $P_2$. In general, a collection of elements $\{P_i\}_{i\in I}\subseteq \mathcal{G}$ determine a context if $P_i\sim P_j$ for all $i,j\in I$. Equivalently, if the subgraph with nodes $\{P_i\}_{i\in I}$ is complete.  A \emph{maximal} context is a context not contained properly in another context.  If we do not indicate the opposite, when we refer to contexts we will be implying maximal contexts.
\end{definition}

To visualize these mathematical definitions we can provide the following elements of some abstract graph $\mathcal{G}$,

\begin{center}
\includegraphics[width=14em]{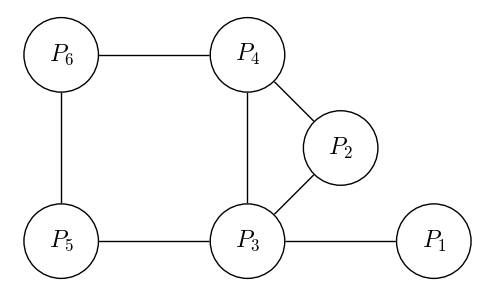}
\end{center}
Notice that $\mathcal{G}$ has much more structure than
just the set of nodes $\{P_1,\ldots,P_6\}$. 
For example, the lines (or edges) between $P_2,P_3$ and $P_4$
indicates that these nodes are mutually related.

If we want to consider a GIV, which is a map $\Psi:\mathcal{G} \to [0,1]$, we have to add the particular intensity of each $P_i \in \mathcal{G}$:

\begin{center}
\includegraphics[width=14em]{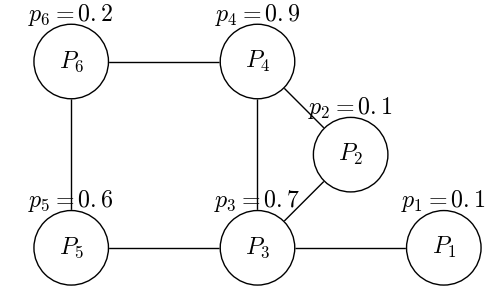}
\end{center}

\noindent For reasons that will become clear in the next sections, we will call the nodes of the graph \emph{powers}. The number over each node is the value under $\Psi$. Note that $\mathcal{G}$ is not complete. An alternative way of representing $\Psi$ is  by using intensive nodes or powers. In this representation we do not write the name of the power neither its exact potentia, 
\begin{center}
\includegraphics[width=14em]{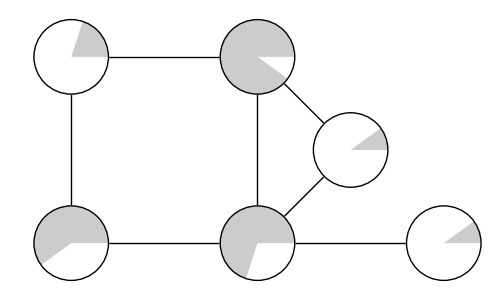}
\end{center}

\noindent With this representation, it is easy to understand
the notion of a context in complete generality.
In the next picture, we circle the maximal contexts.
\begin{center}
\includegraphics[width=14em]{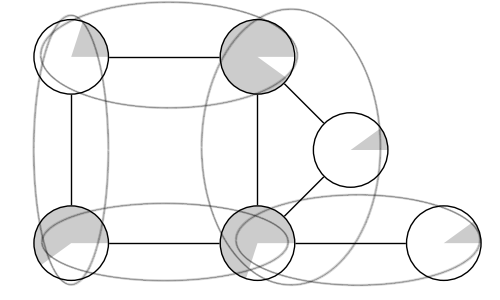}
\end{center}
Notice that the contexts have non-empty intersections.

\begin{example}
Consider again the example of QM.
Let $\mathcal{H}$ be Hilbert space.
Let $\mathcal{G}$ be the graph of immanent powers with the commuting relation 
given by QM
and let $\Psi:\mathcal{G}\rightarrow [0,1]$
be the PSA given through the Born rule with respect to some ray in $\mathcal{H}$.
In this case, the notion of context coincides
with the usual one; a complete set of commuting operators.
\end{example}

As we shall discuss in the next section, our logos approach not only provides a rigorous definition of contextuality which is both intuitive and visualizable, but more importantly, it allows us to escape KS contextuality and restore ---through the intensive valuation of powers to the interval $[0,1]$--- an objective (non-relative) account of the quantum formalism. For some, the price we have paid might seem too high; i.e. the abandonment of classical actualist metaphysics.

\section{The Non-Contextuality Intensive Theorem}

As we mentioned above, the problem of KS contextuality in relation to the definition of physical reality appears from the impossibility to have a ASA as related to the elements that can be possibly measured. Indeed, given a $\Psi$\footnote{We avoid using the term ``quantum state'' for reasons that will become clear in the following.}  the KS-impossibility to consider the state of affairs in terms of definite valued properties, 0 or 1, seemingly precludes an objective (non-relative) representation. However, if we go back to the KS proof discussed in section 4, our {\it global intensive valuation}  is able to provide a non-contextual account of projection operators which can be easily comprehended through the following graph.
\begin{center}
\includegraphics[width=16em]{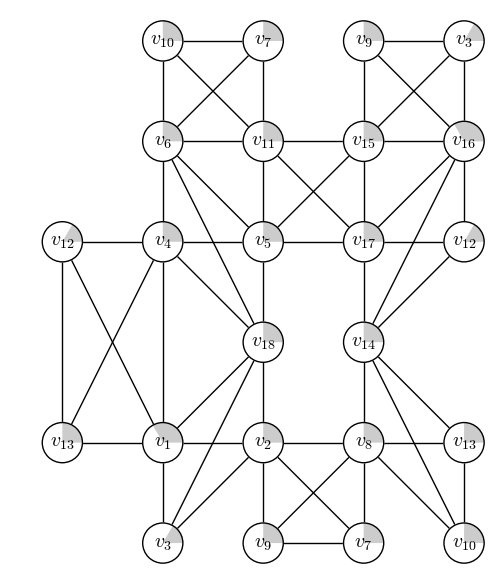}
\end{center}
The intensities drawn are $1/4$, $1/3$ or $1/6$.

By replacing (classical) binary valuations by intensive valuations we are able to relate $\Psi$  to a well defined state of affairs. We are able to now state the following non-contextuality theorem for QM:

\begin{theo} 
{\bf (Intensive Non-Contextuality)} If $\mathcal{H}$ is a Hilbert space, then a 
PSA is possible.
\end{theo} 

\noindent  {\it Proof:} Straightforward, apply the Born rule 
to obtain a PSA, which is a function to  $[0,1]$. \qed \\

\noindent Notice that escaping the metaphysical constraints imposed by the actual realm, our representation and redefinition of how to valuate a quantum wave function, $\Psi$, restores an objective (non-contextual) account of the (intensive) values of all projection operators related to the decomposition of $\Psi$. As we have seen above, our approach can be derived directly from physical considerations (section 6) and the formalism itself in a very natural manner. This also provides a guide regarding the mathematical elements that should be considered in the creation of an objective conceptual representation. 

\begin{coro} 
The non-contextuality intensive theorem restores the possibility of an objective physical representation of any quantum wave function $\Psi$. Contrary to the orthodox interpretation of QM in terms of systems with properties (which imply a binary valuation), our conceptual representation of quantum physical reality is not relative to any particular context, it is global.
\end{coro} 

\noindent  {\it Proof:} Straightforward, for any $\Psi$ there always exists a PSA 
defined over all $\mathcal{G}$. \qed \\

To sum up, we have been able to escape the contextual character of KS theorem, not by changing the formalism in order to restore our classical way of thinking about {\it what there is}, but on the very contrary, by changing the way we think about {\it what there is} in order to restore an objective representation of what QM is really talking about. We are now ready to discuss the conceptual account of our logos approach to QM. From our perspective, the key to finding an adequate representation rests in the possibility to create a set of concepts that match the formalism as well as its main non-classical features. We believe that by adding the conceptual scheme that derives from the notion of {\it immanent intensive power} our logos approach can provide an intuitive understanding of the meaning of intensive valuations. It is through this new conceptual architectonic, which we have discussed extensively in \cite{deRonde16a, deRonde17a, deRonde17b, deRonde17c, deRonde17d}, that we hope to be able to provide a coherent understanding of what QM is really talking about.

\section{An Intuitive Conceptual Understanding of the Logos Approach}

There exists within the literature different viewpoints regarding the problem of understanding QM. The solutions vary with the different standpoints and goals assumed by the different approaches. On the one hand, instrumentalists claim that ``there is nothing to be further explained or understood in QM'' simply because the theory ---considered as an algorithm of prediction--- already provides the correct probabilistic predictions of ``clicks'' in detectors, and that should be the end of the story \cite{FuchsPeres00}. On the other hand, the orthodox perspective within the literature is that the task of philosophers of QM is to ``bridge the gap'' between the quantum formalism and our classical representation of the world \cite{Dorato15}. For example, John Hawthorne \cite[p. 144]{Hawthorne09} argues: ``[A] natural question to ask is how the familiar truths about the macroscopic world that we know and love (`the manifest image') emerge from the ground floor described by the fundamental book of the world. Assuming that we don't wish to concede that most of our ordinary beliefs about the physical world are false, we seem obliged to make the emergence of the familiar world from the ground floor intelligible to ourselves.'' This perspective regarding QM has focused its attention in two main types of reductionistic problems. Firstly, ``limit-problems'' which attempt to explain the path from the quantum formalism to classical physics in terms of a mathematical limit. This is, the well known ``quantum to classical limit'', which can be dissected into more specific problems, such as the basis problem (dealing with the process by which nature chooses one basis rather than other) and the measurement problem (dealing with ``weird'' quantum superpositions and their measurement outcomes). Secondly, there is also a set of ``no-problems''\footnote{I am thankful to Bob Coecke for this linguistic insight. Cagliari, July 2014.} which analyze QM in terms of classical notions, producing an analysis which takes as a standpoint the main concepts of classical physics. These problems ground themselves on the classical account of physical reality and only reflect about the formalism in ``negative terms''; that is, in terms of the failure of QM to account for the classical representation of reality and the use of its concepts: separability, space, time, locality, individuality, identity, actuality, etc. These ``negative'' set of problems are: {\it non-}separability, {\it non-}locality, {\it non-}individuality, {\it non-}identity, etc. To summarize, the orthodox (Bohrian) perspective of ``understanding QM'' necessarily implies a reduction of the theory to our ``common sense'' classical representation of the world. The project is to justify how QM can be related to our already known classical world.  

The classical representation provided by physical theories rests on the existence of physical systems which can be mainly characterized in terms of properties to which one can always apply truth binary valuations. But the orthodox formalism of QM resists ---mainly due to its non-commutative structure--- such an actualist metaphysical interpretation. So far, the conceptual representation imported from classical theories does not seem to work in the quantum domain without either giving up objectivity ---as in the Bohrian scheme \cite{deRonde17b}--- or abandoning the orthodox formalism itself ---like in the case of Bohmian mechanics and GRW, among others. Within the logos approach, we have been able to present a solution which remains ---contrary to the topos approach--- within the gates of the quantum formalism, escaping through {\it intensive valuations} the relativism imposed by {\it binary valuations} on KS analysis ---i.e., the conclusion that quantum reality changes when being observed or not (e.g., see \cite{Butterfield17}). In a single phrase: instead of giving up objectivity, we prefer to give up binary existence. Indeed, our solution is grounded on the redefinition of quantum physical reality beyond the classical (actualist) realm of definite binary valued properties. ``Understanding QM'' implies in our approach to supplement the mathematical formalism with adequate physical concepts that should be capable of explaining not only what the theory talks about but also the experience it presupposes \cite{deRonde16b, deRonde17a}. This means we also need to provide an intuitive grasp of our proposed {\it global intensive valuation}. Without such {\it anschaulich} understanding we will remain only at the level of a ``formal solution'' of KS contextuality. In this same respect, our goal must be to try to find an adequate physical notion which is able to be defined and understood in the intensive terms of a probabilistic measure. This also means, to shift from the classical understanding of {\it classical probability} as `inaccurate knowledge of a state of affairs' to an understanding of {\it quantum probability} in terms of `objective knowledge of a state of affairs' \cite{deRonde16c}. In the previous section we called {\it powers} to each of the elements of $\mathcal{G}$, $P_i$. We believe that this notion can be conceptually developed in order to capture the idea of intensive existence already present within the orthodox quantum formalism \cite{deRonde17a, deRonde17b, deRonde17c, deRonde17d}. 

{\it The notion of immanent power can be understood intuitively as an intensive notion}. One can imagine the existence of a power in relation to a specific potentia which characterizes the intensity or strength of a power.  A power, unlike a property of a system, must be characterized in terms of an intensive existence which measures the potentia of the power. For example, we might argue that Messi has the power of shooting penalties with a degree of accuracy of 0.95. This means that he will score approximately 95 out of 100 penalties. The number 0.95 has an intuitive grasp for anyone who has played a sport. Neymar might have a potentia of shooting penalties of 0.87, and this is why one might better choose Messi to shoot a penalty rather than Neymar. However, if we consider a particular situation Neymar might score while Messi might fail to do so. Particular effectuations do not contradict the statistical causality that we find in probabilistic knowledge. According to our approach, QM describes statistical causality, not particular actual effectuations. 

{\it Our proposed GIV escapes the KS constrain imposed by BV}. The price we are willing to pay in order to restore an objective picture is to give up the consideration of physical reality in terms of (actualist) binary existence. Intensive valuations open the door to a generalized characterization of what can be considered as physically real in truly potential terms. In classical physics, the ontological level has been always exclusively considered in terms of the actual realm. In this case actual observation collapses with the actual mode of existence. The single (actual) observation of a property (in the actual realm of existence) is enough to characterize the property completely. On the contrary, a power existing in the potential realm, cannot be characterized through a single observation. To characterize completely the power $P_i$ we require a statistical measure which indicates its potentia $p_i$. The power possesses an intensive existence which, contrary to classical properties is not either 0 or 1, but a number pertaining to the closed interval $[0,1]$. Thus, to characterize a single power we will require a statistical measure of many actual effectuations. 

{\it The notion of immanent power also captures in a natural way the contextual character of quantum measurements avoiding relativist choices which explicitly change the representation of the state of affairs.} Indeed, in order to measure the potentia of a specific power, we need to prepare a specific context. For example, if we want to measure the power of scoring penalties of Messi or Neymar, then we will obviously need a football field with a goal and a goalkeeper. Notice that this is in no way different to a Stern-Gerlach measurement which also requires the construction of a definite measurement situation. Obviously, if Messi shoots only one penalty we will not get enough information in oder to produce an objective statistical measure of the power in question. Let us also remark that this is completely analogous to the classical case in which objects also assume that a partial perception (or observation) is not enough to characterize the whole. In fact, in each particular observation we only gain access to viewing a partial profile (or {\it adumbration}) of the object in question. For example, if we see a table from above, we will only see `the top of the table'; but we will not be able to infer what type of legs it has. In order to learn what type of legs it has we will obviously need to observe the table from a different angle. Thus, in order to gain a complete knowledge of the whole table we will necessarily require many different observations. As in the case of the table, the measurements we can perform on a specific power change in no way its ontological existence. Each measurement can be regarded as a partial adumbration of the power in question. 

It is also easy to understand through the notion of immanent power the epistemic and ontic aspects of contextuality. For example, the context for scoring penalties is epistemically incompatible to that of throwing corners. Obviously, Messi cannot through a corner and shoot a penalty at the same time. Exactly the same contextual aspect is found in Stern-Gerlach experiments since we cannot measure simultaneously `spin in x-direction' and `spin in the y-direction'. In order to understand more clearly the relation between contextuality, measurements and preexistence, it is important to recall from \cite{deRonde16c} the following distinction:\\

\noindent {\it
{\sc Epistemic Incompatibility of Measurements:} Two contexts are epistemically incompatible if their measurements cannot be performed simultaneously.}\\

\noindent {\it 
{\sc Ontic Incompatibility of Existents:} Two contexts are ontically incompatible if their formal elements cannot be considered as simultaneously preexistent.}

\

\noindent Even though some powers are {\it epistemically incompatible} (i.e., they require mutually incompatible measurement set ups in order to be observed) they are never {\it ontologically incompatible} since they can be all defined to exist simultaneously through our GIV. The {\it Potential State of Affairs}, constituted by the set of potentially existent intensive immanent powers, is in this respect completely objective (i.e., independent of any subjective choice). Notice that this is completely analogous to the way in which the notion of {\it Actual State of Affairs}, as constituted by sets of actually definite valued properties, is also regarded as objective. Objective means in this case that the multiple observations of the same state of affairs are represented in terms of a coherent whole. The main difference between a ASA and a PSA regards the {\it conditions of objectivity}. While in the classical case an ASA is defined in term of a set of systems with definite valued properties; in the quantum case the PSA is defined in intensive terms, through a set of immanent powers with definite potentia. What must be clearly recognized is that in the classical case we are not discussing about the same objects of inquiry as in the quantum case.  

So, while Messi and Neymar possess a list of definite powers or skills, each one of them with a definite potentia, the measurement of each power requires not only a specific context of inquiry, some of which are epistemically incompatible (see for a more detailed analysis: \cite{deRonde16c}); it also requires a statistical measure of each power. In order to visualize our ideas we can present the list of {\it intensive powers} possessed by Messi and Neymar quantified in terms of their precise {\it potentia}.

\begin{footnotesize}
\begin{center}
\begin{tabular}{ | l | c | c |}
\hline
Powers & Messi & Neymar\\ \hline
\hline
Ball Control &0.96&0.88\\ \hline
Dribbling &0.97&0.90\\ \hline
Low Pass &0.94&0.77\\ \hline
Lofted Pass &0.97&0.75\\ \hline
Finishing &0.98&0.86\\ \hline
Header &0.67&0.64\\ \hline
Defensive Prowess &0.45&0.40\\ \hline
Kicking Power &0.88&0.78\\ \hline
Speed &0.95&0.90\\ \hline
Explosive Power &0.98&0.92\\ \hline
Body Control &0.97&0.91\\ \hline
Stamina &0.85&0.80\\ \hline
Injury Resistance &0.7&0.4\\ \hline
\end{tabular}
\end{center}
\end{footnotesize}

\bigskip 

Both the non-contextual existence of the powers $P_i$ and their potentia,  together with the contextual character of measurements of such powers can be clearly visualized through a single graph.\footnote{The fact that even in classical physics we can find epistemically incompatible measurement situations has been discussed by Diederik Aerts in \cite{Aerts82}.}\\ 

\begin{center}
\begin{minipage}{.2\linewidth}
\includegraphics[width=\linewidth]{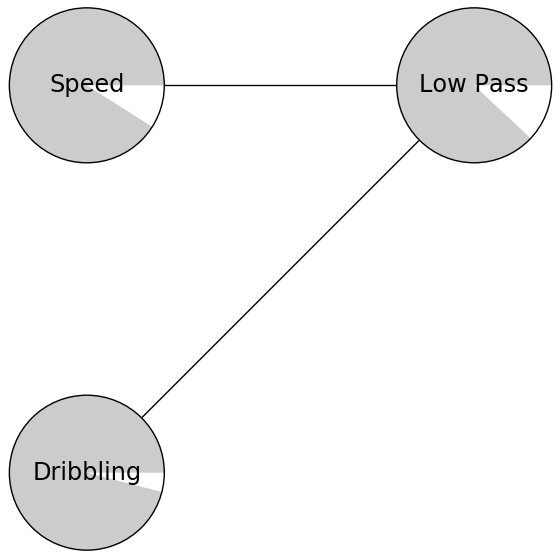}
\end{minipage}%
\qquad
\begin{minipage}{.2\linewidth}
\includegraphics[width=\linewidth]{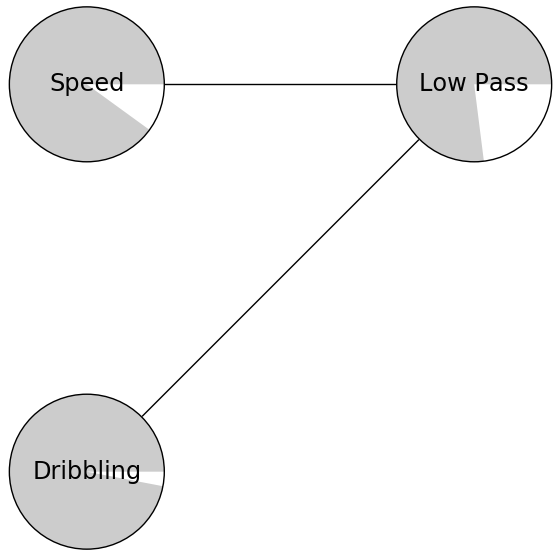}
\end{minipage}
\end{center}

\bigskip 

\noindent Within our approach quantum contexts appear as an epistemic constrain to the possibility of measuring powers simultaneously. Notice that a Stern Gerlach apparatus in a lab can be also considered as a situation where there exists, in the potential realm, a set of ontological non-contextual powers some of which are epistemically incompatible or contextual.\\ 

\begin{center}
\includegraphics[width=9em]{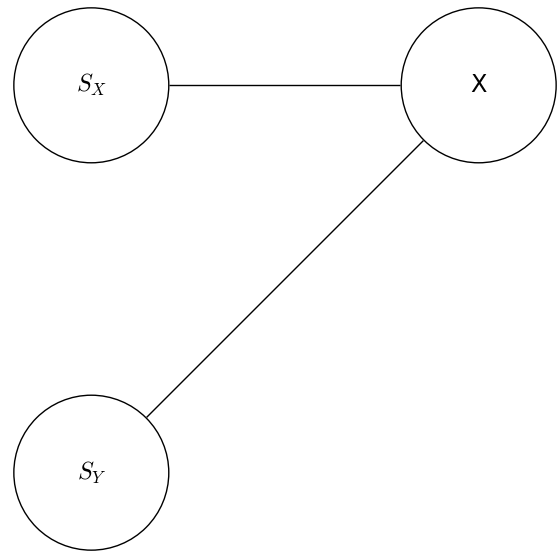}

\end{center} 

\bigskip 

\noindent By complementing our logos approach of {\it global intensive valuations} with the conceptual scheme of {\it powers} and {\it potentia} we are able to present an intuitive grasp of what QM might be talking about. Of course, this is only a first step within a line of research in which there is still much more work to be done in the future.

\section*{Conclusions}

We presented a new categorical approach which captures the contextual nature of QM in a natural manner, allowing us to escape KS contextuality and the relativism imposed by the actualist account of binary valuations. Our non-contextuality intensive theorem explains how KS contextuality can be understood as an epistemic feature of the theory, bypassing the idea that quantum measurements create reality. We have been able to escape the contextual character of KS theorem, not by changing the formalism in order to restore our classical way of thinking about {\it what there is}, but on the very contrary, by changing the way we think about {\it what there is} in order to restore an objective representation of what QM is really talking about. In this respect, we provided an {\it anschaulich} objective account of the logos approach in terms of intensive immanent powers. We have shown how our scheme can be intuitively understood through the aid of a new conceptual scheme grounded on a potential realm of existence completely independent of the actual realm.

\section*{Acknowledgements} 

This work was partially supported by the following grants: FWO project G.0405.08 and FWO-research community W0.030.06. CONICET RES. 4541-12 and the Project PIO-CONICET-UNAJ (15520150100008CO) ``Quantum Superpositions in Quantum Information Processing''.

\end{document}